\documentclass[aps,pra,nofootinbib,notitlepage]{revtex4-1}

\usepackage{fullpage,verbatim}
\usepackage[ruled,vlined]{algorithm2e}
\usepackage{url}           
\usepackage{bm}            
\usepackage{amsthm,amssymb,amsmath}
\usepackage{mathrsfs,setspace,pstcol}
\usepackage[colorlinks]{hyperref}
\usepackage{braket}
\usepackage[nottoc]{tocbibind}

\newcommand\numberthis{\addtocounter{equation}{1}\tag{\theequation}}
\theoremstyle{plain}
\newtheorem{theorem}{Theorem}[section]
\newtheorem{lemma}[theorem]{Lemma}
\newtheorem{fact}{Fact}

\newtheorem{proposition}[theorem]{Proposition}

\theoremstyle{definition}
\newtheorem{definition}[theorem]{Definition}

\newenvironment{Protocol}[1][htb]
  {
   \begin{algorithm}
  }{\end{algorithm}}

\theoremstyle{remark}
\newtheorem{remark}[theorem]{Remark}
\newcommand{\one}{\leavevmode\hbox{\small1\kern-3.8pt\normalsize1}}

\newcommand{\bU}{\mathbf{U}}

\newcommand{\Var}{\mathrm{Var}}
\newcommand{\Tr}{\mathrm{Tr}\,}
\newcommand{\polylog}{\mathrm{polylog}}
\newcommand{\poly}{\mathrm{poly}}
\newcommand{\Haar}{\mathrm{Haar}}
\newcommand{\qTPE}{\mathrm{qTPE}}

\newcommand{\C}{\mathbb{C}}
\newcommand{\E}{\mathbb{E}}
\newcommand{\I}{\mathbb{I}}
\newcommand{\R}{\mathbb{R}}
\newcommand{\U}{\mathbb{U}}
\newcommand{\CP}{\mathbb{CP}}

\newcommand{\cD}{\mathcal{D}}
\newcommand{\cF}{\mathcal{F}}
\newcommand{\cG}{\mathcal{G}}
\newcommand{\cH}{\mathcal{H}}
\newcommand{\cI}{\mathcal{I}}
\newcommand{\cM}{\mathcal{M}}
\newcommand{\cV}{\mathcal{V}}
\newcommand{\cY}{\mathcal{Y}}

\newcommand{\bcF}{\bar{\cF}}

\allowdisplaybreaks

\begin{document} 

\title{
Efficiently estimating 
average fidelity of a quantum logic gate using few classical random bits
} 

\author{Aditya Nema} 
\email{aditya.nema30@gmail.com}
\author{Pranab Sen} 
\email{pgdsen@tcs.tifr.res.in}
\affiliation{
School of Technology and System Science, Tata Institute of 
Fundamental Research, Mumbai, India.
} 

\begin{abstract} 
We give three new algorithms for
efficient in-place estimation, without using ancilla qubits, of 
average fidelity of a quantum logic gate 
\cite{EmersonAlickiZyczkowski}
acting on a $d$ dimensional system using much fewer  random bits than
what was known so far. 
Previous approaches for efficient estimation  of average gate  fidelity
\cite{DankertCleveEmersonLivine} 
replaced Haar random unitaries in the naive estimation algorithm 
by approximate unitary $2$-designs, and sampled them uniformly and
independently. 
In contrast, in our first algorithm
we sample the unitaries of the approximate unitary $2$-design uniformly 
using a limited independence pseudorandom generator, a powerful tool 
from derandomisation theory.
This algorithm uses 
$O(\epsilon^{-2} (\log d) (\log \epsilon^{-1}))$
number of basic operations in order to estimate the average gate fidelity
to within an additive error $\epsilon$, which is the same as 
\cite{DankertCleveEmersonLivine}.
However, it only uses
$O((\log d)(\log \epsilon^{-1}))$ random bits, which  is
much lesser number than the 
$\Omega(\epsilon^{-2} (\log d) (\log \epsilon^{-1}))$
random bits used in the previous works. Reducing the requirement of 
classical random
bits increases the reliability of estimation as often, high quality
random bits are an expensive computational resource.

Our second efficient algorithm, based on a $4$-quantum tensor product
expander, works if the gate dimension $d$ is large. It uses even lesser
random bits than the first algorithm, and has the added advantage that
it needs to implement only one unitary versus potentially all the unitaries
of an approximate $2$-design in the first algorithm.
 Our third efficient algorithm,
based on an $l$-quantum tensor product expander for moderately large
values of $l$, works for all values of the parameters. It uses slightly
more random bits than the other algorithms but has the advantage that
it needs to implement only a small number of unitaries versus potentially 
all the unitaries
of an approximate $2$-design in the first algorithm. This advantage is
of great importance to experimental implementations in the near future.
\end{abstract} 

\maketitle

\section{Introduction} 
Unitary quantum logic gates serve as the basic building
blocks of quantum circuits implementing quantum algorithms. 
They are nothing but unitary operators chosen from a predetermined
set. Implementation of any quantum algorithm is achieved by 
application of an appropriate sequence of these gates. However, in 
practice there is always some error between the ideal gate output given
a particular input state and the actual gate output because of 
noise. Thus if we want to apply a quantum gate, it is 
desirable that the noisy experimental version $\Lambda$ 
be close to the ideal gate $\bU$ with 
respect to some measure. Most works use the so-called {\em fidelity}
as a measure of closeness. This leads to the notion of gate 
fidelity for a particular input state. 
Notice that this characterisation depends on the quantum 
state inputted to the gate. 
However when the gate is used as part of a quantum
circuit, it may not be feasible to figure out the states that may possibly
be inputted to the gate during the course of operation of the circuit.
One would like to remove this state dependence in the definition
of gate fidelity and instead come up with a quantity that
serves as a general measure of the quality of the gate implementation. 
One way to do this is to consider the gate fidelity averaged over the
Haar probability measure on all possible pure input states. This quantity
is called the
{\em average gate fidelity} \cite{EmersonAlickiZyczkowski}. 

Emerson et al. \cite{EmersonAlickiZyczkowski} gave an algorithm 
for estimating average fidelity of a unitary
$d$-dimensional gate using several independent samples of 
$d \times d$ Haar random 
unitaries. While they did not explicitly bound the number of samples
required in order to obtain an estimate of the average fidelity to
within additive error $\epsilon$, their method can be analysed to show
that $O(\epsilon^{-2})$ independent Haar random samples of unitaries
suffice. This is prohibitively expensive, both
in terms of the computational cost required to implement the Haar random
unitaries as well as in terms of the number of
random bits required to do the sampling. Both these quantities are
at least $\Omega(\epsilon^{-2} \frac{d^2}{\log d} \log (1/\epsilon))$
\cite[Corollary~4.2.13]{Vershynin}. 
Later works \cite{DankertCleveEmersonLivine, EmersonEtAl} showed that 
the computational cost and the number of random bits 
can be drastically brought down 
by choosing independent
uniform samples from a unitary $2$-design. Again, these papers did 
not rigorously
estimate the running time and the usage of random bits of their algorithm,
but nevertheless those can be analysed to obtain a bound of
$O(\epsilon^{-2} \polylog(d/\epsilon))$ for both.

In this paper, we treat the usage of classical random bits as a resource
and try to minimise it. One of the reasons behind this is that high
quality random bits are an expensive computational resource. Thus,
reducing them while preserving the efficiency of the algorithm and
the estimation error would generally lead to more reliable estimation
in practice. This is especially important because in the near future,
we expect to have experimental realisations of quantum circuits on
about ten to fifty qubits, and they have to be benchmarked accurately
and reliably so as to further the goal for having a practical quantum
computer one day. Another reason behind reducing the number of random
bits is that it often gives us deep insights into the algorithms
involved, leading to serendipitous additional optimisations that may
have otherwise escaped our attention. In this paper, we shall actually
see an example of such a serendipitous optimisation. Our second and
third algorithms not only use far less random bits that what was known
earlier; they also require the algorithms to implement a much smaller
set of unitaries than potentially all unitaries of an approximate
$2$-design in earlier works
\cite{DankertCleveEmersonLivine, EmersonEtAl}. This is very important
for the near future because implementing a wide range of unitaries is
technologically fraught with immense challenges. An excellent introduction
to derandomisation goals and techniques in computer science can
be found in the book \cite{Goldreich}.

We give three new algorithms for
efficient in-place estimation, without using ancilla qubits, of the
average fidelity of a quantum logic gate
acting on a $d$ dimensional system using much fewer  random bits than
what was known so far. 
We consider in-place algorithms only because good quality qubits are 
likely to remain a very expensive resource in near term experimental
implementations, and so we want to avoid ancilla qubits as much as
possible.
In our first algorithm, in contrast to earlier works,
we sample the unitaries of the approximate unitary $2$-design uniformly 
using a limited independence pseudorandom generator
\cite{AlonGoldreichHastadPeralta, SchmidtSiegelSrinivasan}, a powerful 
tool from derandomisation theory.
This algorithm uses 
$O(\epsilon^{-2} (\log \delta^{-1}) (\log d) (\log \epsilon^{-1}))$
number of basic operations in order to estimate the average gate fidelity
to within an additive error $\epsilon$ with confidence $1-\delta$, 
which is the same as 
\cite{DankertCleveEmersonLivine}.
However, it only uses
$
O((\log \delta^{-1})(
(\log d)(\log \epsilon^{-1}) +
\log \log \delta^{-1}
)
)
$ 
random bits, which  is
much lesser number than the 
$\Omega(\epsilon^{-2} (\log d) (\log \epsilon^{-1}))$
random bits used in the previous works. 

Our second efficient algorithm, based on a $4$-quantum tensor product
expander \cite{HarrowHastings}, works if the gate dimension $d$ is 
large. It uses even lesser
random bits than the first algorithm, and has the added advantage that
it needs to implement only one unitary versus potentially all 
the unitaries
of an approximate $2$-design in the first algorithm.
Our third efficient algorithm,
based on an $l$-quantum tensor product expander \cite{HarrowHastings} 
for moderately large
values of $l$, works for all values of the parameters. It uses slightly
more random bits than the other algorithms but has the advantage that
it needs to implement only a small number of unitaries versus potentially 
all the unitaries
of an approximate $2$-design in the first algorithm. 

The rest of the 
paper is organised as follows: 
\begin{itemize}

\item 
In Section~\ref{sec:prelim} we build up some notations 
and preliminaries that will 
be used throughout the paper.

\item 
Section~\ref{sec:tail} describes how to bound the tail of
the gate fidelity distribution when the unitaries are chosen from
an approximate $l$-quantum tensor product expander (qTPE).

\item 
Section~\ref{sec:dankert} summarises the earlier work on estimating
average gate fidelity adding a rigorous analysis of estimation and
confidence errors missing in previous works.

\item 
In Section~\ref{sec:efficient}, we describe a randomness efficient
algorithm for estimating average gate fidelity using approximate
$2$-designs combined with a limited independence pseudorandom 
generator.

\item 
In Section~\ref{sec:4design}, we describe a randomness efficient
algorithm for estimating average gate fidelity using an approximate
$4$-qTPE, which works if the gate dimension is large and needs to
implement only one unitary. 

\item 
In Section~\ref{sec:qTPE}, we describe a randomness efficient
algorithm for estimating average gate fidelity using approximate
$l$-designs for moderate values of $l \geq 4$, which has the advantage
that the estimation procedure needs to potentially implement only a
very small number of unitaries.

\item 
We conclude in Section~\ref{sec:conclude} with a discussion
of what we have achieved, and directions for future work.

\end{itemize}
    
\section{Notation and preliminaries}
\label{sec:prelim}
Throughout the paper, $\cH$ denotes a complex Hilbert 
space of finite dimension $d$ and $\cH^{\otimes m}$ denotes the $m$ 
fold tensor product of $\cH$. We use $\one$ to denote the 
identity operator on $\cH$. $\cM_{k,d}$ denotes the 
vector space of $k \times d$ linear operators over complex field and 
$\cM_d = \cM_{d,d}$. Note that $\cM_d$ is itself 
a Hilbert dpace of dimension $d^2$ with Hilbert-Schmidt inner product, 
defined as: 
$ 
\langle A,B \rangle \triangleq Tr(A^{\dagger}B).
$ 
For $p > 0$, we let
$\lVert \cdot \rVert_p$ denote the Schatten $p$-norm
of operators in $\cM_d$, defined as: 
$
\lVert A \rVert_p \triangleq 
[\Tr (A^\dagger A)^{p/2}]^{1/p}.
$ 
This is nothing but the $\ell_p$-norm of the vector of singular values
of $A$. The case $p = 1$ is called the {\em trace norm}. The case $p = 2$
is the Frobenius norm or the Hilbert-Schmidt norm induced from the
eponymous inner product. Letting $p \rightarrow \infty$ gives us
the Schatten $\ell_\infty$-norm which is nothing but the largest singular
value of $A$, aka operator norm or spectral norm of $A$.
Often, we use $\rho$ to denote a quantum state or a density matrix which is
a Hermitian, positive semidefinite matrix with unit trace. We let
$\cD(d)$ denote the set of all $d \times d$ density 
matrices. A pure quantum quantum is a rank one density matrix. We
denote by $\CP^{d-1}$ the set of pure quantum states in $\cH$.
The notation $\ket{\cdot}$ denotes a vector of unit 
$\ell_2$-length. Thus if $\ket{\psi}$ is a unit vector, 
$\ket{\psi}\bra{\psi}$ is a pure quantum state. We will often abuse
notation and use $\ket{\psi}$ to denote a pure quantum state also.

We use $\Lambda$ 
to denote the noisy experimental realisation of an ideal unitary quantum 
gate $\bU$. For a bipartite Hilbert space
$\cH_1 \otimes \cH_2$, the partial trace $\Tr_{\cH_2}$ denotes the 
operation of 
tracing out $\cH_2$. We use $X$ to 
denote a random variable and $\bar{X}$ to denote its expected value with
respect to a probability measure $\mu$, i.e., 
$\bar{X} \triangleq \int X \, d\mu$. The notation $\Var(X)$ denotes the 
variance of random variable $X$, i.e. $Var(X) = 
\overline{(X-\bar{X})^2}$. 
The symbol $\U(d)$ stands for the 
the unitary group on $\cH$ i.e. the group of $d \times d$ complex
unitary matrices. 
We tacitly assume that the
ceiling is taken of any formula that provides dimension or value of $t$ in 
unitary $t$-design. The symbol $\Haar$ is used to denote the
unique unitarily invariant 
Haar probability measure on $\U(d)$, or $\CP^{d-1}$ as appropriate. 
Expectation with respect to a measure $\mu$ is denoted by $\E_\nu[\cdot]$.

{\em Fidelity} between two quantum states $\rho$ and $\sigma$ is 
defined as: 
$
F(\rho,\sigma) \triangleq 
\lVert \sqrt{\rho}\sqrt{\sigma} \rVert_1^2 = 
{(\Tr\sqrt{\sqrt{\rho}\sigma\sqrt{\rho}} )}^2.
$ 
Fidelity is a measure of distinguishabilty of two states. 
It is easy to see that
$F(\rho,\sigma)=1$ 
implies $\rho$ and $\sigma$ are identical, and $F(\rho,\sigma)=0$ implies 
that $\rho$ and $\sigma$ have orthogonal support and there exists a single 
measurement that distinguishes them perfectly. Fidelity is related to
trace distance via the following inequality.
\begin{equation}
\label{eq:fidelitytracedistance}
1 - \sqrt{F(\rho,\sigma)} \leq
\frac{\lVert \rho - \sigma \rVert_1}{2} \leq
\sqrt{1 - F(\rho, \sigma)}.
\end{equation}

A linear 
mapping $\Lambda: \mathcal{M}_m \to \mathcal{M}_d $ is called a super 
operator, and a super operator which is completely positive and trace 
preserving is considered as a quantum operation. The vector space of
superoperators is denoted by $L(\cM_m, \cM_d)$ or just $L(\cM_m)$ if
$m = d$.
Let $\bU \in \U(d)$. Then 
$\bU$ is also a quantum operation defined 
as $\bU(\rho)= U \rho U^{\dagger}$. 
Suppose $\Lambda$ is a `noisy'
implementation of the unitary quantum operation $\bU$. 
Then 
$\Lambda$ is a quantum operation from $\cM_d$ to $\cM_d$, and so it can 
be represented using 
{\em Kraus operators} as
$
\Lambda(\rho) = \sum_k A_k \rho A_k^{\dagger},
$
where $\{A_k\}_k$ are $d \times d$ matrices called Kraus operators 
of $\Lambda$, with the property that 
$\Sigma_k A_k^{\dagger} A_k = \one$. It turns out that $d^2$ Kraus
operators suffice to describe any quantum operation from $\cM_d$ to 
$\cM_d$. We shall measure the distance between two superoperators
via the so-called {\em diamond norm} \cite{KitaevWatrous}. The 
diamond norm of a superoperator
$\Lambda: \cM_d \rightarrow \cM_d$ is defined as follows:
\[
\lVert \Lambda \rVert_\Diamond :=
\sup_m \max_{X \in \cM_{dm}: \lVert X \rVert_1 = 1} 
\lVert (\I_m \otimes \Lambda)(X) \rVert_1 =
\max_{X \in \cM_{d^2}: \lVert X \rVert_1 = 1} 
\lVert (\I_d \otimes \Lambda)(X) \rVert_1,
\]
where $\I_m$ is the identity superoperator on $\cM_m$.
A quantum operation $\Lambda$ always has 
$\lVert \Lambda \rVert_\Diamond = 1$.

{\em Gate fidelity} between 
$\Lambda$ and $\bU$ for an input state $\rho$ is defined as: 
$$ 
\cF_{\Lambda,\bU}(\rho) \triangleq 
F(\Lambda(\rho),\bU(\rho)) =
(\Tr\sqrt{\sqrt{\Lambda(\rho)}\bU(\rho)\sqrt{\Lambda(\rho)}})^2.
$$ 
The average gate fidelity $\bcF_{\Lambda, \bU}$ is now defined to be the 
expectation of the
gate fidelity $\cF_{\Lambda, \bU}(\ket{\psi}\bra{\psi})$ for pure input 
states $\ket{\psi}\bra{\psi}$ chosen from the Haar
probability measure on $\CP^{d-1}$:
\[
\bcF_{\Lambda, \bU} :=
\int_{\CP^{d-1}} 
(\Tr\sqrt{\sqrt{\Lambda(\rho)}\bU(\rho)\sqrt{\Lambda(\rho)}})^2 \;
d\,\Haar(\psi)
\]

In practice, when one wants to benchmark the quality
of the experimental implementation of a unitary quantum logic gate $\bU$,
one runs the implementation twice, first in the forward direction followed
by the backward direction. If the implemenation were indeed perfect,
this would just do the quantum operation $\bU^{-1} \bU = \one$. 
But because the implementation is noisy what we get is a quantum
operation, which we will again denote by $\Lambda$, that is close
to the identity quantum operation. Let $\ket{\psi}$ be a unit length
vector in $\cH$. Define
\[
\cF_{\Lambda}(\ket{\psi}) \triangleq
\bra{\psi} \Lambda(\ket{\psi}\bra{\psi}) \ket{\psi}.
\]
Alternately, let $V \in \U(d)$. Define
\[
\cF_{\Lambda}(V) \triangleq
\bra{0} V^{-1} \Lambda(V \ket{0}\bra{0} V^{-1}) V \ket{0}.
\]
Then the average gate fidelity of
$\Lambda$ is given by
\[
\bcF_{\Lambda} \triangleq
\bcF_{\Lambda, \one} =
\int_{\CP^{d-1}} \bra{\psi} \Lambda(\ket{\psi}\bra{\psi}) \ket{\psi}
d\,\Haar(\psi) =
\int_{\U(d)} \cF_{\Lambda}(V) d\,\Haar(V).
\]
The average gate fidelity defined above has a nice expression 
in terms of
the Kraus operators of $\Lambda$ \cite{EmersonAlickiZyczkowski}:
\[
\bcF_{\Lambda} =
\sum_k \frac{|\Tr A_k|^2 + d}{d^2 + d}.
\]

The variance of the gate fidelity under the Haar measure on $d \times d$
unitaries $V$ happens to satisfy the following inequality
\cite[Equation~18]{MagesanBlumekohoutEmerson}.
\begin{equation}
\label{eq:variance}
\Var_V[\cF_{\Lambda}(V)] \leq \frac{26}{d}.
\end{equation}

The gate fidelity is an example of a function from the unit 
$\ell_2$-norm sphere $S^{2d-1}$ in $\C^d$ to $\R$. More generally, 
consider a function $f: (S^{2d-1})^{\times t} \rightarrow \R$ defined on a
direct product of $t$ spheres. Let $\eta > 0$. The function $f$ 
is said to be {\em $\eta$-Lipshitz} if for all unit length vectors
$v_1, \ldots v_t, w_1, \ldots, w_t \in S^{2d-1}$, we have
\[
|f(v_1, \ldots v_t) - f(w_1, \ldots, w_t)| \leq
\eta \sqrt{\sum_{i=1}^t \lVert v_i - w_i \rVert_2^2}.
\]
The following fact was proved in \cite[Lemma~2.7]{FawziHaydenSen}.
\begin{fact}
\label{fact:conc}
Let $\eta > 0$.
Let $f: (S^{2d-1})^{\times t} \rightarrow \R$ be $\eta$-Lipschitz.
Consider the probability distribution on points 
$(v_1, \ldots, v_t) \in (S^{2d-1})^{\times t}$ where 
the $v_i$s are chosen independently from the Haar measure on
$S^{2d-1}$. Define 
\[
\mu \triangleq 
\E_{(v_1, \ldots, v_t)}[f(v_1, \ldots, v_t)].
\]
Let $\delta > 0$. Then,
\[
\Pr_{(v_1, \ldots, v_t)}[
|f(v_1, \ldots, v_t) - \mu| > \delta
] \leq
4 \exp(-\frac{\delta^2 d}{16 \eta^2}).
\]
\end{fact}

Now let $f$ be the following function
\[
f(\ket{\psi_1}, \ldots, \ket{\psi_t}) \triangleq
t^{-1} \sum_{i=1}^t F_{\Lambda}(\ket{\psi_i}).
\]
Trivially, 
\[
\mu \triangleq 
\E_{\ket{\psi_1}, \ldots, \ket{\psi_t}}[
f(\ket{\psi_1}, \ldots, \ket{\psi_t})
] =
\bcF_{\Lambda}.
\]
It is easy to see that 
\begin{eqnarray*}
\lefteqn{
|
f(\ket{\psi_1}, \ldots, \ket{\psi_t}) -
f(\ket{\phi_1}, \ldots, \ket{\phi_t})
|
} \\
& = &
t^{-1} |
\sum_{i=1}^t (\cF_{\Lambda}(\ket{\psi_i}) - \cF_{\Lambda}(\ket{\phi_i}))
| 
\;\leq\;
t^{-1} \sum_{i=1}^t 
|\cF_{\Lambda}(\ket{\psi_i}) - \cF_{\Lambda}(\ket{\phi_i})| \\
& \leq &
t^{-1} 
\sum_{i=1}^t (
|
\Tr[\Lambda(\ket{\psi_i}\bra{\psi_i}) \ket{\psi_i}\bra{\psi_i}] -
\Tr[\Lambda(\ket{\psi_i}\bra{\psi_i}) \ket{\phi_i}\bra{\phi_i}]
| +
|
\Tr[\Lambda(\ket{\psi_i}\bra{\psi_i}) \ket{\phi_i}\bra{\phi_i}] -
\Tr[\Lambda(\ket{\phi_i}\bra{\phi_i}) \ket{\phi_i}\bra{\phi_i}]
|
) \\
& \leq &
t^{-1} 
\sum_{i=1}^t (
\lVert \ket{\psi_i}\bra{\psi_i}] - \ket{\phi_i}\bra{\phi_i} \rVert_1
+
\lVert 
\Lambda(\ket{\psi_i}\bra{\psi_i}) -
\Lambda(\ket{\phi_i}\bra{\phi_i}) 
\rVert_1
) \\
& \leq &
2 t^{-1} \sum_{i=1}^t 
\lVert \ket{\psi_i}\bra{\psi_i} - \ket{\phi_i}\bra{\phi_i} \rVert_1 
\;\leq\;
4 t^{-1} \sum_{i=1}^t 
\lVert \ket{\psi_i} - \ket{\phi_i} \rVert_2 
\;\leq\;
4 t^{-1/2} 
\sqrt{
\sum_{i=1}^t \lVert \ket{\psi_i} - \ket{\phi_i} \rVert_2^2 
},
\end{eqnarray*}
which shows that the Lipschitz constant $\eta \leq 4 t^{-1/2}$.
Let $\delta > 0$. Consider the probability distribution on
points $(V_1, \ldots, V_t) \in \U(d)^{\times t}$ obtained by choosing
each $V_i$ independently from the Haar measure. Define 
$f(V_1, \ldots, V_t)$ in the natural fashion.
By Fact~\ref{fact:conc}, 
\begin{equation}
\label{eq:conc}
\Pr_{(V_1, \ldots, V_t)}[
|f(V_1, \ldots, V_t) - \mu| > \delta
] \leq
4 \exp(-\frac{\delta^2 d t}{256}).
\end{equation}

Now, we define an approximate 
unitary $2$-design via the so-called {\em twirling operation} as 
in \cite{DankertCleveEmersonLivine}. 
Let $\nu$ be a 
probability measure on $\U(d)$. Let $\Lambda$ be a superoperator
on $\cM_d$.
Define the $\nu$-twirling 
operation $E_{\nu}: L(\cM_d) \rightarrow L(\cM_d)$ as follows:  
\begin{equation} 
\label{eq:gnu} 
E_\nu(\Lambda) \triangleq
\left(
X \rightarrow 
\int_{\U(d)} V^\dagger \Lambda(V X V^\dagger) V \; d\nu(V)
\right),
\end{equation} 
where $\Lambda \in L(\cM_d)$ and $X \in \cM_d$.
When $\nu$ is the Haar probability measure on $\U(d)$, we shall write
the above superoperator as $E_{\Haar}(\Lambda)$.
\begin{definition}
The probablity measure  $\nu$ is an $\epsilon$-approximate unitary 
$2$-design if:
\begin{equation} 
\label{eq:approxdesign}
\lVert E_\nu(\Lambda) - E_{\Haar}(\Lambda) \rVert_\Diamond 
\leq 
\epsilon \lVert \Lambda \rVert_{\Diamond}
\end{equation}
for all superoperators $\Lambda \in L(\cM_d)$.
If $\epsilon = 0$, then $\nu$ is said to be an exact unitary $2$-design.
\end{definition}

We now recall the definition of a quantum tensor power expander
\cite{HarrowHastings}.
\begin{definition}
A quantum $t$-tensor product expander ($t$-qTPE) in $\cH$, $|\cH| = d$, of
degree $s$ can be defined as a quantum operation
$
\cG: L(\cH^{\otimes t}) \rightarrow L(\cH^{\otimes t})
$
that can be expressed as
$
\cG(M) = 
\frac{1}{s} \sum_{i=1}^s (U_i)^{\otimes t} M (U_i^{-1})^{\otimes t},
$
for any matrix $M \in L(\cH^{\otimes t})$, where
$\{U_i\}_{i=1}^s$ are $d \times d$ unitary matrices. The 
qTPE is said to have second singular value $\lambda$ if
$
\lVert \cG - \cI \rVert_\infty \leq \lambda,
$
where $\cI$ is the `ideal' quantum operation defined by its action on
a matrix $M$ by
$
\cI(M) := 
\int_{U \in \U(D)} U^{\otimes t} M (U^\dag)^{\otimes t} \,
d\,\Haar(U).
$
In other words, if $M \in L(\cH^{\otimes t})$, then
$
\lVert \cG(M) - \cI(M) \rVert_2 \leq \lambda \lVert M \rVert_2.
$
We use  the notation $(d, s, \lambda, t)$-qTPE to denote such a
quantum tensor product expander.
\end{definition}
From the above definition, it is easy to see that a
$(d, s, \lambda, t)$-qTPE is also simultaneously a
$(d, s, \lambda, t')$-qTPE for any $t' < t$.

Let $X_1, X_2, \ldots, X_n$ be a sequence of $\{0,1\}$-valued random 
variables. 
Let $2 \leq k \leq n$. 
Let $S \subseteq \{1, 2, \ldots, n\}$, $S = \{s_1, s_2, \ldots, s_k\}$.
Let $X_{s_1} X_{s_2} \cdots X_{s_k}$ denote the actual joint distribution
of the corresponding random variables. Let $0 < p < 1$. Let
$B(k, p)$ denote the Bernoulli distribution i.e. the distribution
of $k$ fully independent identical coin tosses with probability of a
coin turning up HEAD equal to $p$.
The sequence $X_1, X_2, \ldots, X_n$  is said to be 
{\em $\theta$-approximate $p$-biased $k$-wise independent} if for any
subset $S = \{s_1, \ldots, s_k\}$, 
\[
\lVert 
X_{s_1} X_{s_2} \cdots X_{s_k} -
B(k, p)
\rVert_1 \leq \theta,
\]
i.e. the joint probability distribution of any subset of the random
variables of size $k$ is $\theta$-close to the Bernoulli distribution
in $\ell_1$-distance.

We first state a Chernoff bound for a fully independent 
(i.e. $0$-approximate $n$-wise independent) sequence of random
variables. This bound can be derived from
\cite[Corollary~A.1.14]{AlonSpencer}.
\begin{fact}
\label{fact:StandardChernoff}
Let $X_1, X_2, \ldots, X_n$ be a fully independent sequence of 
identially distributed $\{0,1\}$-valued random variables. Let
$p \triangleq \E[X]$. Let $0 < \epsilon < p$. 
Let $X \triangleq n^{-1} \sum_{i=1}^n X_i$. Then
\[
\Pr [|X - p| > \epsilon] \leq 2 \exp(-\frac{\epsilon^2 n}{3}).
\]
\end{fact}
We now state a Chernoff-like bound for $\theta$-approximte $k$-wise
independent random variables.  This bound can be derived from
\cite[Theorem~5]{SchmidtSiegelSrinivasan}.
\begin{fact}
\label{fact:LimitedChernoff}
Let $X_1, X_2, \ldots, X_n$ be a $\theta$-approximate $p$-biased $k$-wise
independent sequence of $\{0,1\}$-valued random variables.
Let $0 < \epsilon < p$. 
Let $X \triangleq n^{-1} \sum_{i=1}^n X_i$. Suppose 
$k = e^{-1/3} \epsilon^2 n$. Then,
\[
\Pr [|X - p| > \epsilon] \leq   
\exp(-k/2) + \theta (\frac{n}{\epsilon})^k.
\]
\end{fact}

We now recall that a $\theta$-approximate $1/2$-biased $k$-wise 
independent sequence of random bits can be efficiently constructed by 
spending only a small amount of truly random bits 
\cite[Theorem~3]{AlonGoldreichHastadPeralta}.
\begin{fact}
\label{fact:smallbias}
Let $k$, $n$ be positive integers, $\theta > 0$ and
$r \triangleq k + 2 \log \log k + 2 \log \log n + 2 \log \theta^{-1}$.
Then there is a function $f: \{0,1\}^r \rightarrow \{0,1\}^n$ such
that, if the uniform distribution on $\{0,1\}^r$ is provided at the input,
the resulting sequence $X_1, X_2, \ldots, X_n$ at 
the output is a $\theta$-approximate $1/2$-biased $k$-wise 
independent sequence of random bits. 
Moreover there is a deterministic
algorithm that, given an input string $z \in \{0,1\}^r$ and a bit 
position $i \in \{1, 2, \ldots, n\}$, 
computes the output bit $f(z)_i$ in time $\poly(r)$.
\end{fact}
Facts~\ref{fact:LimitedChernoff} and \ref{fact:smallbias} allow us 
to prove the following easily.
\begin{fact}
\label{fact:sampling}
Let $\cY$ be a set and $g$ a function 
$g: \cY \rightarrow [0,1]$. Let $p \triangleq \E_{Y}[g(Y)]$, where the
expectation is taken over the random variable $Y$ uniformly distributed 
on $\cY$. 
Let $n$ be a positive integer, $\epsilon, \theta > 0$ and
$
r \triangleq 
4 \epsilon^2 n \log |\cY|) + 2 \log \theta^{-1}.
$
Then there is a function $f: \{0,1\}^r \rightarrow \cY^n$ such
that, if the uniform distribution on $\{0,1\}^r$ is provided at the input,
the resulting sequence $X_1, X_2, \ldots, X_n$ of random bits at 
the output, where $X_i = 1$ with probability $g(Y_i)$ and $0$ otherwise,
satisfies
\[
\Pr [|X - p| > \epsilon] \leq   
\exp(-\epsilon^2 n / 4) + 
\theta (\frac{n}{\epsilon})^{\frac{\epsilon^2 n}{4}},
\]
$X$ being defined as $n^{-1} \sum_{i=1}^n X_i$. 
Moreover there is a deterministic
algorithm that, given an input string $z \in \{0,1\}^r$ and a 
position $i \in \{1, 2, \ldots, n\}$, 
computes the output element $f(z)_i \in \cY$ in time 
$\poly(r) \log |\cY|$.
\end{fact}

\section{Bounding the tail of the gate fidelity distribution}
\label{sec:tail}
In this section, we show how to bound the tail of the gate fidelity
distribution, both under the Haar measure as well as under the
uniform measure on a qTPE. As a warmup, we first show that
approximate unitary $2$-designs and $2$-qTPEs are related.
\begin{fact}
\label{fact:TPEvsDesign}
A $(d, s, \lambda, 2)$-qTPE is a $(\lambda d^4)$-approximate
unitary $2$-design consisting of $s$ unitaries. 
\end{fact}
\begin{proof}
Let $S$
denote the swap operator on $\cH^{\otimes 2} \otimes \cH^{\otimes 2}$ 
which swaps the two multiplicands of the central tensor product symbol. 
Swap is a unitary operation.
Let $V$ be a unitary operator, and $X$ a linear operator 
on $\cH \otimes \cH$. Let $\one$ denote the identity operator on $\cH$.
Let $\I$ be the identity superoperator on $L(\cH)$. 
Let $\Lambda: L(\cH) \rightarrow L(\cH)$ be a superoperator. Then,
\begin{eqnarray*}
\lefteqn{
\left|
\Tr [
(
(
s^{-1} \sum_{i=1}^s 
(U_i^{-1} \otimes \one)
((\Lambda \otimes \I)((U_i \otimes \one) X (U_i^{-1} \otimes \one))) 
(U_i \otimes \one)
)
\right.
} \\
\lefteqn{
~~~~~~~~
\left.
{} -
(
\int_{\U(d)} 
(U^{-1} \otimes \one)
((\Lambda \otimes \I)((U \otimes \one) X (U^{-1} \otimes \one))) 
(U \otimes \one) \, d\,\Haar(U)
)
) V
] 
\right|
} \\
& = &
\left|
s^{-1} \sum_{i=1}^s 
\Tr [
((\Lambda \otimes \I)((U_i \otimes \one) X (U_i^{-1} \otimes \one))) 
(U_i \otimes \one) V (U_i^{-1} \otimes \one)
] 
\right. \\
&    &
~~~~
\left.
{} -
\int_{\U(d)} 
\Tr [
((\Lambda \otimes \I)((U \otimes \one) X (U^{-1} \otimes \one))) 
(U \otimes \one) V (U^{-1} \otimes \one)
] \, d\,\Haar(U)
\right| \\
& = &
\left|
s^{-1} \sum_{i=1}^s 
\Tr [
(
((\Lambda \otimes \I)((U_i \otimes \one) X (U_i^{-1} \otimes \one))) 
\otimes
((U_i \otimes \one) V (U_i^{-1} \otimes \one))
) S
] 
\right. \\
&    &
~~~~
\left.
{} -
\int_{\U(d)} 
\Tr [
(
((\Lambda \otimes \I)((U \otimes \one) X (U^{-1} \otimes \one))) 
\otimes 
((U \otimes \one) V (U^{-1} \otimes \one))
) S
] \, d\,\Haar(U)
\right| \\
& = &
\left|
\Tr [
(
((\Lambda \otimes \I) \otimes (\I \otimes \I))
(
s^{-1} \sum_{i=1}^s 
((U_i \otimes \one) \otimes (U_i \otimes \one))
(X \otimes V) 
((U_i^{-1} \otimes \one) \otimes (U_i^{-1} \otimes \one))
)
) S
] 
\right. \\
&    &
~~~~
\left.
{} -
\Tr [
(
((\Lambda \otimes \I) \otimes (\I \otimes \I))
(
\int_{\U(d)} 
((U \otimes \one) \otimes (U \otimes \one))
(X \otimes V) 
((U^{-1} \otimes \one) \otimes (U^{-1} \otimes \one))\, 
d\,\Haar(U)
) 
) S
] 
\right| \\
& = &
|
\Tr [
(
((\Lambda \otimes \I) \otimes (\I \otimes \I)) 
(
(
(\cG \otimes' (\I \otimes'' \I)) -
(\cI \otimes' (\I \otimes'' \I))
)
(X \otimes V)
)
) S
]
| \\
& \leq &
\lVert
((\Lambda \otimes \I) \otimes (\I \otimes \I)) 
(
(
(\cG \otimes' (\I \otimes'' \I)) -
(\cI \otimes' (\I \otimes'' \I))
)
(X \otimes V)
)
\rVert_1 \\
& \leq &
\lVert \Lambda \rVert_\Diamond \cdot
\lVert
(
(\cG \otimes' (\I \otimes'' \I)) -
(\cI \otimes' (\I \otimes'' \I))
)
(X \otimes V)
\rVert_1 \\
& \leq &
d^2 \lVert \Lambda \rVert_\Diamond \cdot
\lVert X \otimes V \rVert_1 \cdot
\lVert
(\cG \otimes' (\I \otimes'' \I)) -
(\cI \otimes' (\I \otimes'' \I))
\rVert_\infty \\
&   =  &
d^4 \lVert \Lambda \rVert_\Diamond 
\lVert X \rVert_1 \cdot
\lVert \cG - \cI \rVert_\infty 
\;  = \;
\lambda d^4 \lVert \Lambda \rVert_\Diamond 
\lVert X \rVert_1.
\end{eqnarray*}
Above, we used the so-called {\em swap trick} 
$\Tr [M N] = \Tr [(M \otimes N) S]$ in the second equality, 
$\otimes'$, $\otimes''$ in the fourth equality indicate that the 
splitting of the tensor multiplicands is different from the splitting
in $X \otimes V$, the fact that 
$
\lVert \cG(X) \rVert_1 \leq 
\sqrt{|\cH|} \lVert X \rVert_1 \lVert \cG \rVert_\infty
$
for an operator $X \in L(\cH)$ and superoperator 
$\cG: L(\cH) \rightarrow L(\cH)$
in the third inequality, and $\lVert V \rVert_1 = d^2$ in the fifth
equality. Since 
\begin{eqnarray*}
\lefteqn{
\lVert
E_\cG(\Lambda) - E_{\Haar}(\Lambda)
\rVert_\Diamond
} \\
&  = &
\max_{X: \lVert X \rVert_1 = 1}
\lVert
((E_\cG(\Lambda) - E_{\Haar}(\Lambda)) \otimes \I)(X)
\rVert_1 \\
&  = &
\max_{X: \lVert X \rVert_1 = 1}
\max_{V: \mbox{unitary}}
|
\Tr [
(((E_\cG(\Lambda) - E_{\Haar}(\Lambda)) \otimes \I)(X)) V
]
| \\
&  = &
\max_{X: \lVert X \rVert_1 = 1}
\max_{V: \mbox{unitary}}
\left|
\Tr [
(
(
s^{-1} \sum_{i=1}^s 
(U_i^{-1} \otimes \one)
((\Lambda \otimes \I)((U_i \otimes \one) X (U_i^{-1} \otimes \one))) 
(U_i \otimes \one)
)
\right. \\
&   &
~~~~~~~~~~~~~~~~~~~~~~~~~~~~~~~~~~~~~
\left.
{} -
(
\int_{\U(d)} 
(U^{-1} \otimes \one)
((\Lambda \otimes \I)((U \otimes \one) X (U^{-1} \otimes \one))) 
(U \otimes \one) \, d\,\Haar(U)
)
) V
] 
\right|,
\end{eqnarray*}
we get
$
\lVert
E_\cG(\Lambda) - E_{\Haar}(\Lambda)
\rVert_\Diamond \leq 
\lambda d^4 \lVert \Lambda \rVert_\Diamond.
$
This completes the proof.
\end{proof}

We now prove four lemmas that will help us relate the tail  of the
distribution of gate fidelity calculated with respect to a unitary
chosen from the Haar measure versus chosen from a $t$-qTPE.
\begin{lemma} 
\label{lem:Kraus}
Let $\{A_k\}_k$ be Kraus operators 
of quantum operation $\Lambda$, i.e., 
$\Lambda(\rho) = \sum_k A_k \rho A_k^{\dagger}$. 
We will use $\cF(V)$ as a shorthand for $\cF_{\Lambda}(V)$. Let
$l$ be a positive integer. Let $a \in \C$. Define 
$
M \triangleq 
(\sum_k A_k \otimes A_k^\dagger) -
a \one_{d^2}.
$
Then, 
\[
\cF(V) - a = 
\Tr[
(
V^{\dagger})^{\otimes 2}
M
V^{\otimes 2}
)
\ket{00}\bra{00}
],
\]
and
$
(\cF(V) - a)^l = 
\Tr[
(
(V^{\dagger})^{\otimes 2l}
M^{\otimes l}
V^{\otimes 2l}
)
\ket{0^{2l}}\bra{0^{2l}}
].
$
\end{lemma}
\begin{proof}
We have
\begin{align*}
\cF(V) - a
& = 
\bra{0} V^{-1} \Lambda(V \ket{0}\bra{0} V^{-1}) V \ket{0} - a
\; = 
\bra{0} V^{\dagger} 
(
\sum_k A_k V \ket{0}\bra{0} V^{\dagger} A_k^{\dagger} 
)
V \ket{0} - a \\
& = 
(
\sum_k 
\bra{0} V^{\dagger} A_k V \ket{0} \cdot
\bra{0} V^{\dagger} A_k^{\dagger} V \ket{0} 
) - a
\; = 
(
\sum_k 
\Tr[(V^{\dagger} A_k V) \ket{0}\bra{0}]
\Tr[(V^{\dagger} A_k^\dagger V) \ket{0}\bra{0}] 
) - a \\
& = 
(
\sum_k 
\Tr[
((V^{\dagger} A_k V) \ket{0}\bra{0}) \otimes 
((V^{\dagger} A_k^\dagger V) \ket{0}\bra{0})
]
) - a \\
& = 
(
\sum_k 
\Tr[
(V^{\dagger} \otimes V^{\dagger}) 
(A_k \otimes A_k^{\dagger})
(V \otimes V) 
(\ket{0} \otimes \ket{0})
(\bra{0} \otimes \bra{0})
]
) - a \\ 
& = 
\sum_k 
\Tr[
(V^{\dagger} \otimes V^{\dagger}) 
(A_k \otimes A_k^{\dagger})
(V \otimes V) 
\ket{00}\bra{00}
] - a \Tr[\one_{d^2} \ket{00}\bra{00}]
\; =
\Tr[
(V^{\dagger})^{\otimes 2} M V^{\otimes 2}
\ket{00}\bra{00}
].
\end{align*}
This proves the first equality. 
For the second equality, 
\begin{eqnarray*}
(\cF(V) - a)^l 
& = &
(
\Tr[
(V^{\dagger})^{\otimes 2} M V^{\otimes 2}
\ket{00}\bra{00}
]
)^l 
\;=\;
\Tr[
(
(V^{\dagger})^{\otimes 2} M V^{\otimes 2}
\ket{00}\bra{00}
)^{\otimes l}
] \\
& = &
\Tr[
(V^{\dagger})^{\otimes 2l} M^{\otimes l} V^{\otimes 2l}
\ket{0^{2l}}\bra{0^{2l}}
].
\end{eqnarray*}
This completes the proof.
\end{proof}
\begin{lemma} 
\label{lem:Krausbound}
Under the notation of Lemma~\ref{lem:Kraus},
$\lVert M \rVert_2 \leq (1+|a|)d$ and 
$\lVert M^{\otimes l} \rVert_2 \leq ((1+|a|)d)^l$.
\end{lemma}
\begin{proof}
Define $N \triangleq \sum_k A_k \otimes A_k^\dagger$.
We have
\begin{align*}
\lVert N \rVert_2 
& = 
\sqrt{\Tr[N N^{\dagger}]}
\;=
\sqrt{
\Tr[
\sum_{j,k}
(A_j \otimes A_j^{\dagger})(A_k^{\dagger} \otimes A_k)
]
} \\
& =
\sqrt{
\sum_{j,k} 
\Tr[A_j A_k^{\dagger}] 
\Tr[A_j^{\dagger} A_k]
} 
\;=
\sqrt{
\sum_{j,k} \Tr[A_j^\dagger A_k] \overline{\Tr[A_j^\dagger A_k]}
}\\
& =
\sqrt{
\sum_{j,k}
\lvert \langle A_j, A_k \rangle \rvert^2 
}
\;\leq \sqrt{\sum_{j,k} \lVert A_j \rVert_2^2 \lVert A_k \rVert_2^2}
\;=
\sum_{k} \lVert A_k \rVert_2^2
\;=
\Tr[\sum_k A_k^\dagger A_k] 
\;=
\Tr[\one_{d}] 
\;=
d,
\end{align*}
where the inequality is obtained by applying Cauchy-Schwarz to 
the Hilbert-Schmidt inner product. Then 
\[
\lVert M \rVert_2 \leq
\lVert N \rVert_2 + |a| \lVert \one_{d^2} \rVert_2 =
(1+|a|) d.
\]
This completes the proof of the lemma.
\end{proof} 
\begin{lemma} 
\label{lem:qTPEMoment}
Let $d$, $l$, $s$ be positive integers and $\lambda > 0$. Let $a \in \C$.
The notation
$\E_{V: \qTPE}[\cdot]$ denotes the expectation with respect to a
$d \times d$ unitary $V$ chosen uniformly at random from a
$(d, s, \lambda, 2l)$-qTPE.
The notation
$\E_{V: \Haar}[\cdot]$ denotes the expectation with respect to a
$d \times d$ unitary $V$ chosen from the Haar measure. Then,
\[
\lvert 
\E_{V: \qTPE}[(\cF(V) - a)^l] -
\E_{V: \Haar}[(\cF(V) - a)^l] 
\rvert \leq 
\lambda ((1+|a|)d)^l.
\]
\end{lemma}
\begin{proof}
Using Lemmas~\ref{lem:Kraus} and \ref{lem:Krausbound}, we get
\begin{eqnarray*}
\lefteqn{
\lvert 
\E_{V: \qTPE}[(\cF(V) - a)^l] -
\E_{V: \Haar}[(\cF(V) - a)^l] 
\rvert
} \\
& = &
|
\E_{V: \qTPE}[
\Tr[
(
(V^{\dagger})^{\otimes 2l}
M^{\otimes l}
V^{\otimes 2l}
)
\ket{0^{2l}}\bra{0^{2l}}
]
] -
\E_{V: \Haar}[
\Tr[
(
(V^{\dagger})^{\otimes 2l}
M^{\otimes l}
V^{\otimes 2l}
)
\ket{0^{2l}}\bra{0^{2l}}
]
]
| \\
& = &
|
\Tr[
(
\E_{V: \qTPE}[
(V^{\dagger})^{\otimes 2l}
M^{\otimes l}
V^{\otimes 2l}
] -
\E_{V: \Haar}[
(V^{\dagger})^{\otimes 2l}
M^{\otimes l}
V^{\otimes 2l}
] 
)
\ket{0^{2l}}\bra{0^{2l}}
] \\
& \leq &
\lVert
\E_{V: \qTPE}[
(V^{\dagger})^{\otimes 2l}
M^{\otimes l}
V^{\otimes 2l}
] -
\E_{V: \Haar}[
(V^{\dagger})^{\otimes 2l}
M^{\otimes l}
V^{\otimes 2l}
] 
\rVert_2 \\
& \leq &
\lambda \lvert M^{\otimes l} \rVert_2
\;\leq\;
\lambda ((1+|a|)d)^l.
\end{eqnarray*}
This completes the proof of the lemma.
\end{proof}
\begin{lemma} 
\label{lem:qTPEMomentAverage}
Let $t$ be a positive integer. 
Let $d$, $l$, $s$ be positive integers and $\lambda > 0$. Let 
$0 \leq a \leq 1$.
The notation
$\E_{V_1, \ldots, V_t: \qTPE}[\cdot]$ denotes the expectation with 
respect to independently choosing 
$d \times d$ unitaries $V_1, \ldots, V_t$ uniformly at random from a
$(d, s, \lambda, 2l)$-qTPE.
The notation
$\E_{V_1, \ldots, V_t: \Haar}[\cdot]$ denotes the expectation with 
respect to independently choosing
$d \times d$ unitaries $V_1, \ldots, V_t$ from the  Haar measure. 
Define the function $f: \U(d)^{\times t} \rightarrow \R$ as
\[
f(V_1, \ldots, V_t) \triangleq
t^{-1} \sum_{i=1}^t \cF_{\Lambda}(V_i),
\]
i.e., $f$ is the average of $t$ gate fidelities.
Then,
\[
|
\E_{(V_1, \ldots, V_t): \qTPE}[
(f(V_1, \ldots, V_t) - a)^{l}
] -
\E_{(V_1, \ldots, V_t): \Haar}[
(f(V_1, \ldots, V_t) - a)^{l}
] 
| \leq 
\lambda (2d)^l.
\]
\end{lemma}
\begin{proof}
From Lemma~\ref{lem:qTPEMoment}, it is easy to see that
\begin{eqnarray*}
\lefteqn{
|
\E_{(V_1, \ldots, V_t): \qTPE}[
(f(V_1, \ldots, V_t) - a)^{l}
] -
\E_{(V_1, \ldots, V_t): \Haar}[
(f(V_1, \ldots, V_t) - a)^{l}
] 
|
} \\
&   =  &
t^{-l} |
\E_{(V_1, \ldots, V_t): \qTPE}[
(\sum_{i=1}^t \cF_{\Lambda}(V_i) - a)^{l}
] -
\E_{(V_1, \ldots, V_t): \Haar}[
(\sum_{i=1}^t \cF_{\Lambda}(V_i) - a)^{l}
] 
| \\
& \leq &
t^{-l} 
\sum_{i_1, \ldots, i_t: \sum_{j=1}^t i_j = l}
{l \choose i_1 \cdots i_t}
|
\prod_{j=1}^t
\E_{V_j: \qTPE}[
(\cF_{\Lambda}(V_j) - a)^{i_j}
] -
\prod_{j=1}^t
\E_{V_j: \Haar}[
(\cF_{\Lambda}(V_j) - a)^{i_j}
]
| \\
& \leq &
t^{-l} 
\sum_{i_1, \ldots, i_t: \sum_{j=1}^t i_j = l}
{l \choose i_1 \cdots i_t}
\sum_{m=t}^1
|
\prod_{j=1}^{m}
\E_{V_j: \qTPE}[
(\cF_{\Lambda}(V_j) - a)^{i_j}
] 
\prod_{j=m+1}^t
\E_{V_j: \Haar}[
(\cF_{\Lambda}(V_j) - a)^{i_j}
] \\
&     &
~~~~~~~~~~~~~~~~~~~~~~~~~~~~~~~~~~~~~~~~~~~~~~
{} -
\prod_{j=1}^{m-1}
\E_{V_j: \qTPE}[
(\cF_{\Lambda}(V_j) - a)^{i_j}
] 
\prod_{j=m}^t
\E_{V_j: \Haar}[
(\cF_{\Lambda}(V_j) - a)^{i_j}
]
| \\
&   =  &  
t^{-l} 
\sum_{i_1, \ldots, i_t: \sum_{j=1}^t i_j = l}
{l \choose i_1 \cdots i_t}
\sum_{m=t}^1
|
\prod_{j=1}^{m-1}
\E_{V_j: \qTPE}[
(\cF_{\Lambda}(V_j) - a)^{i_j}
] 
\prod_{j=m+1}^t
\E_{V_j: \Haar}[
(\cF_{\Lambda}(V_j) - a)^{i_j}
] \\
&    &
~~~~~~~~~~~~~~~~~~~~~~~~~~~~~~~~~~~~~~~~~~~~~~~~~~~~~~~~
(
\E_{V_m: \qTPE}[
(\cF_{\Lambda}(V_j) - a)^{i_m}
] -
\E_{V_m: \Haar}[
(\cF_{\Lambda}(V_j) - a)^{i_m}
]
)
| \\
& \leq &  
t^{-l} 
\sum_{i_1, \ldots, i_t: \sum_{j=1}^t i_j = l}
{l \choose i_1 \cdots i_t}
\sum_{m=t}^1
|
(
\E_{V_m: \qTPE}[
(\cF_{\Lambda}(V_j) - a)^{i_m}
] -
\E_{V_m: \Haar}[
(\cF_{\Lambda}(V_j) - a)^{i_m}
]
)
| \\
& \leq &  
t^{-l} \lambda ((1+a)d)^{l}
\sum_{i_1, \ldots, i_t: \sum_{j=1}^t i_j = l}
{l \choose i_1 \cdots i_t} \\
& \leq &
\lambda (2d)^{2l},
\end{eqnarray*}
where in the fourth inequality we used the fact that if 
$(i_1, \ldots, i_p)$ is a partition of
$l$ where each $i_m \neq 0$, $1 \leq m \leq p$,  then for any $x > 1$,
$
\sum_{m=1}^p x^{i_m} \leq 
\prod_{m=1}^p x^{i_m} \leq x^l. 
$
This completes the proof of the lemma.
\end{proof}

We now prove the following important result giving a tail bound
on the uniform average of $t$ gate fidelity functions,
when the $t$ unitaries are chosen independently and uniformly from
a qTPE.
\begin{proposition}
\label{prop:conc}
Let $t$, $d$, $s$, $l$ be  positive integers and $\delta, \lambda > 0$.
Consider the probability distribution on points
$(V_1, \ldots, V_t) \in \U(d)^{\times t}$ obtained by choosing
each $V_i$ independently and uniformly from a $(d, s, \lambda, 4l)$-qTPE.
Define the function $f: \U(d)^{\times t} \rightarrow \R$ as
\[
f(V_1, \ldots, V_t) \triangleq
t^{-1} \sum_{i=1}^t \cF_{\Lambda}(V_i),
\]
i.e., $f$ is the average of $t$ gate fidelities.
Then,
\[
\Pr_{(V_1, \ldots, V_t): \qTPE}[
|f(V_1, \ldots, V_t) - \bcF_{\Lambda}| > \delta
] \leq
\delta^{-2l} (
4 (\frac{256 l}{d t})^l + 
\lambda (2 d)^{2l}
).
\]
For the special case where $l = 1$, we get
\[
\Pr_{(V_1, \ldots, V_t): \qTPE}[
|f(V_1, \ldots, V_t) - \bcF_{\Lambda}| > \delta
] \leq
\delta^{-2} (
\frac{26}{d t} + 
\lambda (2 d)^{2}
).
\]
\end{proposition}
\begin{proof}
By Lemma~\ref{lem:qTPEMomentAverage},
\[
|
\E_{(V_1, \ldots, V_t): \qTPE}[
(f(V_1, \ldots, V_t) - \bcF_{\Lambda})^{l}
] -
\E_{(V_1, \ldots, V_t): \Haar}[
(f(V_1, \ldots, V_t) - \bcF_{\Lambda})^{l}
] 
| \leq 
\lambda (2d)^l.
\]
From Equation~\ref{eq:variance}, we get that
\[
\E_{(V_1, \ldots, V_t): \Haar}[
(f(V_1, \ldots, V_t) - \bcF_{\Lambda})^{2}
] =
t^{-1} \E_{V: \Haar}[
(\cF_{\Lambda}(V) - \bcF_{\Lambda})^{2}
] =
\frac{26}{d t}.
\]

Observe that
\begin{eqnarray*}
\lefteqn{
\Pr_{(V_1, \ldots, V_t): \qTPE}[
|f(V_1, \ldots, V_t) - \bcF_{\Lambda}| > \delta
] 
} \\
& \leq &
\delta^{-2l}
\E_{(V_1, \ldots, V_t): \qTPE}[
(f(V_1, \ldots, V_t) - \bcF_{\Lambda})^{2l}
] 
\;\leq\;
\delta^{-2l} (
\E_{(V_1, \ldots, V_t): \Haar}[
(f(V_1, \ldots, V_t) - \bcF_{\Lambda})^{2l}
] + 
\lambda (2 d)^{2l}
).
\end{eqnarray*}

Now for $l=1$ we have
\[
\Pr_{(V_1, \ldots, V_t): \qTPE}[
|f(V_1, \ldots, V_t) - \bcF_{\Lambda}| > \delta
] \leq 
\delta^{-2l} (
\frac{26}{d t} + 
\lambda (2 d)^{2}
).
\]
For larger values of $l$, we employ the $l$-moment method of
\cite{BellareRompel} as adapted into the quantum setting by
Low~\cite{Low}, combined with the tail bound of Equation~\ref{eq:conc}
for the Haar measure. We obtain
\[
\E_{(V_1, \ldots, V_t): \Haar}[
(f(V_1, \ldots, V_t) - \bcF_{\Lambda})^{2l}
] \leq
4 (\frac{256 l}{d t})^l. 
\]
This gives us
\[
\Pr_{(V_1, \ldots, V_t): \qTPE}[
|f(V_1, \ldots, V_t) - \bcF_{\Lambda}| > \delta
] \leq 
\delta^{-2l} (
4 (\frac{256 l}{d t})^l + 
\lambda (2 d)^{2l}
).
\]
This completes the proof of the proposition.
\end{proof}

\section{Earlier work on estimating average gate fidelity}
\label{sec:dankert}
We first recall the naive algorithm \cite{EmersonAlickiZyczkowski} 
for estimating average gate fidelity
using Haar random unitaries for reference and comparison. The basic
procedure is the following. 
\begin{algorithm}[!!!hhh]
\label{algo:naivebasicprocedure}
\SetAlgoLined
\caption{Basic procedure of naive algorithm}
\begin{enumerate}

\item 
Start with the state $\ket{0} \in \mathcal{H}$;

\item 
Apply a $d \times d$ Haar-random unitary matrix 
$V \in \U(d)$ on $\cH$; 

\item 
Apply the quantum operation $\Lambda$ (the experimental realisation
of $U^{-1} U$) to the state obtained in the above step;

\item 
Apply $V^{-1}$ to the state obtained in the above step;

\item 
Measure the resulting state according to the binary outcome measurement
$\{\ket{0}\bra{0}, \one_{\cH} - \ket{0}\bra{0}\}$. Declare success if
the outcome $\ket{0}\bra{0}$ is observed.
\end{enumerate}
\end{algorithm}
It is easy to see that the probability of success in one iteration of
the basic procedure is, by taking $\ket{\psi} \triangleq U \ket{0}$,
\[
\int_{\U(d)} 
\bra{0} V^{-1} \Lambda(V \ket{0}\bra{0} V^{-1}) V \ket{0} \, 
d\,\Haar(V) =
\int_{\CP^{d-1}}
\bra{\psi} \Lambda(\ket{\psi}\bra{\psi}) \ket{\psi} \, 
d\,\Haar(\psi) =
\bcF_{\Lambda}.
\]
The basic procedure of the naive algorithm is prohibitively expensive, both
in terms of the computational cost required to implement the Haar random
unitaries as well as in terms of the number of
random bits required to do the sampling. Uniformly sampling a 
$d \times d$ Haar 
random unitary to within $\ell_2$-distance $\epsilon$ requires 
at least $\Omega(d^2 \log (1/\epsilon))$ random bits
and circuit size 
at least $\Omega(\frac{d^2}{\log d} \log (1/\epsilon))$
\cite[Lemma~3.5]{Vershynin}. Thus the overall number of random bits used
by the basic procedure
becomes  $\Omega(d^2 \log (1/\epsilon))$, and the overall circuit size
of the basic procedure
$\Omega(\frac{d^2}{\log d} \log (1/\epsilon))$.

Since twirling with a Haar random unitary is the same as twirling with
a uniformly random unitary chosen from an exact unitary $2$-design, 
we can replace the Haar random unitary in the basic procedure 
by a uniformly
random unitary chosen from an exact $2$-design without changing the
probability of success at all \cite{DankertCleveEmersonLivine}. 
The advantage of doing so is that there exist $2$-designs each 
of whose unitaries
can be implemented by a circuit of size $O((\log d)^2)$, whereas a
Haar random unitary almost always requires circuits of size
at least $\Omega(d^2 \log (1/\epsilon))$ to implement within precision
$\epsilon$.

One iteration of the basic procedure succeeds with probability
$\bcF_{\Lambda}$. A single outcome, success or failure, gives us no
clue about the value of $\bcF_{\Lambda}$. So in order to 
actually estimate $\bcF_{\Lambda}$
to within an additive error $\epsilon$, with confidence $1 - \delta$, 
we have to repeat the basic procedure several times. 
Neither \cite{EmersonAlickiZyczkowski} nor 
\cite{DankertCleveEmersonLivine} do this rigorously. We now address
this important shortcoming.
Let us repeat the basic 
procedure independently $\Theta(\epsilon^{-2} \log (1/\delta))$ times 
and take the empirical average 
of successes. By Fact~\ref{fact:StandardChernoff},
an estimate of $\bcF$ to within an additive error of $\epsilon$ with
probability at least $1 - \delta$.
The running time of this algorithm turns out to be
$O(\epsilon^{-2} \log \delta^{-1} (\log d)^2)$, and the number of
random bits consumed turns out to be 
$O(\epsilon^{-2} \log \delta^{-1} \log \epsilon^{-1} (\log d)^8)$
\cite{DiVincenzoLeungTerhal}.
We note that there are constructions of exact unitary $2$-designs
\cite{CleveLeungLiuWang} using Clifford gate circuits of size
$O((\log d) (\log \log d)^2 (\log \log \log d))$, but they use
at least $\Omega(\log d)$ ancilla qubits. Since qubits are likely
to be a very expensive resource in foreseeable implementations of
quantum gates, we prefer that all our algorithms be in-place without
using ancilla qubits. 

Dankert et al.~\cite{DankertCleveEmersonLivine} showed how to improve
the running time and the number of random bits in the basic procedure
by replacing the use of an exact $2$-design by a $\theta$-approximate
$2$-design. Due to the replacement, the probability of success in 
the basic procedure becomes 
\[
\int_{\U(d)} 
\bra{0} V^{-1} \Lambda(V \ket{0}\bra{0} V^{-1}) V \ket{0} \, 
d\nu(V) \leq
\int_{\U(d)} 
\bra{0} V^{-1} \Lambda(V \ket{0}\bra{0} V^{-1}) V \ket{0} \, 
d\,\Haar(V) + 
\theta \lVert \Lambda \rVert_\Diamond =
\bcF_{\Lambda} + \theta,
\]
where we used the fact that $\lVert \Lambda \rVert_\Diamond = 1$
as $\Lambda$ is a quantum operation. We can now take 
$\theta = \epsilon / 2$ and repeat the basic procedure 
$\Theta(\epsilon^{-2} \log (1/\delta))$ times and take the 
empirical average of successes, in order to get an estimate of
$\bcF_{\Lambda}$ to within an additive error of $\epsilon$ with
confidence $1 - \delta$.
The running time and usage of random bits of this algorithm turn out to be
$O(\epsilon^{-2} \log \delta^{-1} (\log d) \log \epsilon^{-1})$.
\cite{DankertCleveEmersonLivine}.
We note that Dankert et 
al.~\cite[Theorem~3]{DankertCleveEmersonLivine} 
incorrectly repeat the basic procedure $O(\log \delta^{-1})$ times to
get confidence $1 - \delta$, ignoring the issue of additive error
$\epsilon$ completely.

By Fact~\ref{fact:TPEvsDesign},
replacing the Haar random unitary $U$ in the
basic procedure by a $(d, s, \frac{\epsilon}{2d^4}, 2)$-qTPE $\cG$ 
also leads to an 
efficient algorithm for estimating average gate fidelity. In fact,
a direct analysis shows that a $(d, s, \frac{\epsilon}{2d}, 2)$-qTPE
suffices too.
Such $2$-qTPEs can be obtained via the so-called {\em zigzag product} 
\cite{Sen}.
The unitaries of the qTPE can be implemented by circuits of size
$O((\log d)^2 (\log d/\epsilon)$. The number of random bits required
is only $O(\log (d/\epsilon))$.
The running time of this algorithm turns out to be
$O(\epsilon^{-2} \log \delta^{-1} (\log d)^2 (\log d/\epsilon))$, and 
the number of
random bits consumed turns out to be 
$O(\epsilon^{-2} \log \delta^{-1} \log (d/\epsilon))$.
The circuit size is slightly inferior to 
Dankert et al.~\cite{DankertCleveEmersonLivine} but the number of random
bits used is less.

\section{Randomness efficient algorithm using approximate unitary 
$2$-designs}
\label{sec:efficient}
Instead of repeating Algorithm~\ref{algo:naivebasicprocedure} with 
independently chosen
unitaries per iteration, we pick the sequence of unitaries of the 
approximate
$2$-design from an approximate $k$-wise independent distribution for
a suitable value of $k$. For clarity, we give the full algorithm below.
\begin{algorithm}[!!!hhh]
\label{algo:efficientbasicprocedure}
\SetAlgoLined
\caption{Basic procedure of randomness efficient algorithm}
{\bf Input:} Classical description of a $d \times d$ unitary $Y$.

\begin{enumerate}

\item 
Start with the state $\ket{0} \in \mathcal{H}$;

\item 
Apply $Y$ on $\cH$; 

\item 
Apply the quantum operation $\Lambda$ (the experimental realisation
of $U^{-1} U$) to the state obtained in the above step;

\item 
Apply $Y^{-1}$ to the state obtained in the above step;

\item 
Measure the resulting state according to the binary outcome measurement
$\{\ket{0}\bra{0}, \one_{\cH} - \ket{0}\bra{0}\}$. Declare success if
the outcome $\ket{0}\bra{0}$ is observed.
\end{enumerate}
\end{algorithm}
\begin{algorithm}[!!!hhh]
\label{algo:efficient}
\SetAlgoLined
\caption{Randomness efficient algorithm using approximate $2$-design}
{\bf Input:} $\epsilon, \delta > 0$.

\ 

{\bf Assumption:} $\epsilon < \bcF_{\Lambda}$.

\ 

{\bf Define:} $n \triangleq \frac{2^4 \log (2/\delta)}{\epsilon^2}$, 
$
\theta \triangleq 
\frac{\delta}{2} (\frac{\epsilon}{n})^{\frac{\epsilon^2 n}{4}}.
$

\ 

{\bf Take:} set $\cY$ to be an $\epsilon/2$-approximate unitary 
$2$-design with
$\log |\cY| = O(\log (d/\epsilon))$. A unitary of $\cY$ can be
implemented by circuits of size $O((\log d)^2 (\log (d/\epsilon)))$. 
This follows from the so-called zigzag product \cite{Sen}.

\ 

{\bf Construct:} a sequence $Y_1, Y_2, \ldots, Y_n$ of unitaries from
$\cY$ as the output of $f: \{0,1\}^r \rightarrow \cY^n$ guaranteed
by Fact~\ref{fact:sampling}, when a uniformly random input
$z \in \{0,1\}^r$ is fed to $f$, where
\[
r \triangleq
4 \epsilon^2 (n \log |\cY|) + 2 \log \theta^{-1} =
O(\log \delta^{-1} (\log (d/\epsilon) + \log \log \delta^{-1})). 
\]

\ 

{\bf Run:} Algorithm~\ref{algo:efficientbasicprocedure} with unitaries
$Y_1, Y_2, \ldots, Y_n$. Record the outputs $b_1, b_2, \ldots, b_n$.
Declare $b := n^{-1} \sum_{i=1}^n b_i$ as the estimate for 
$\bcF_{\Lambda}$.
\end{algorithm}

Algorithm~\ref{algo:efficient} consists of a
classical preprocessing step where the sequence $Y_1, Y_2, \ldots, Y_n$
is computed. This takes classical deterministic time
$O(\epsilon^{-2} \poly(r))$. After the preprocessing step,
Algorithm~\ref{algo:efficient} runs in quantum time
$O(\epsilon^{-2} \log \delta^{-1} (\log d)^2 (\log (d/\epsilon)))$.
The number of random bits used is
$
r = 
O(\log \delta^{-1} (\log (d/\epsilon) + \log \log \delta^{-1})). 
$
By Fact~\ref{fact:sampling}, Algorithm~\ref{algo:efficient} gives
an estimate $b$ of $\bcF_{\Lambda}$ such that
$
\Pr [|b - \bcF_{\Lambda}| > \epsilon] \leq \delta. 
$
Comparing with the zigzag product based algorithm in 
Section~\ref{sec:dankert},
we get the same quantum circuit size but much lesser usage of
random bits. 

Instead of the zigzag product, we can take
the $\epsilon/2$-approximate unitary $2$-design of
Dankert et al. \cite{DankertCleveEmersonLivine}. 
An advantage of doing this is that we use only Clifford
gates in the approximate unitary $2$-design which may be technologically
easier to implement. 
That would give
us classical preprocessing time of $O(\epsilon^{-2} \poly(r))$,
quantum time of 
$O(\epsilon^{-2} \log \delta^{-1} (\log d) (\log \epsilon^{-1}))$ and
number of random bits
$
r =
O(\log \delta^{-1} ((\log d) (\log \epsilon^{-1}) 
  + \log \log \delta^{-1})). 
$
Comparing with the corresponding algorithm in
Section~\ref{sec:dankert},
we get the same quantum circuit size but much lesser usage of
random bits. Thus this algorithm is a good candidate for actual 
experimental implementations in the near future.

\section{Randomness efficient algorithm using approximate 
$4$-qTPE}
\label{sec:4design}
Suppose $\frac{108}{\epsilon^2 d} < \frac{\delta}{2}$. Then there is an
even more randomness efficient algorithm than the one given in
Section~\ref{sec:efficient} as follows.
\begin{algorithm}[!!!hhh]
\label{algo:4design}
\SetAlgoLined
\caption{Randomness efficient algorithm using approximate $4$-qTPE}
{\bf Input:} $\epsilon, \delta > 0$.

\ 

{\bf Assumption:} 
$\epsilon < \frac{\bcF_{\Lambda}}{2}$,
$\frac{108}{\epsilon^2 d} < \frac{\delta}{2}$.

\ 

{\bf Define:} $n \triangleq \frac{12 \log (4 \delta^{-1})}{\epsilon^2}$. 

\ 

{\bf Take:} set $\cY$ to be a $(d, s, \frac{1}{4 d^3}, 4)$-qTPE
with $\log |\cY| = O(\log d)$. A unitary of $\cY$ can be
implemented by circuits of size $O((\log d)^3)$.
This follows from the so-called zigzag product \cite{Sen}.

\ 

{\bf Choose:} a uniformly random unitary $Y$ from $\cY$.

\ 

{\bf Run:} Algorithm~\ref{algo:efficientbasicprocedure} $n$ times
with the same unitary $Y$.
Record the outputs $b_1, b_2, \ldots, b_n$.
Declare $b := n^{-1} \sum_{i=1}^n b_i$ as the estimate for 
$\bcF_{\Lambda}$.
\end{algorithm}

Using Proposition~\ref{prop:conc} in the special case where 
$t = l = 1$ and $\lambda = \frac{1}{4 d^3}$, we get
\[
\Pr_{Y: \qTPE}[
|\cF_{\Lambda}(Y) - \bcF_{\Lambda}| > \epsilon/2
] \leq
4 \epsilon^{-2} (
\frac{26}{d} + 
\lambda (2 d)^{2}
) \leq
\frac{108}{\epsilon^2 d} \leq
\frac{\delta}{2}.
\]
Sampling uniformly from the qTPE $\cY$ requires only $O(\log d)$
random bits. Each unitary of $\cY$ can be implemented in quantum 
time $O((\log d)^3)$. These facts follow from the so-called zigzag
product \cite{Sen}.

Now suppose that
$
|\cF_{\Lambda}(Y) - \bcF_{\Lambda}| < \frac{\epsilon}{2}
$
indeed. 
By Fact~\ref{fact:StandardChernoff}, running 
Algorithm~\ref{algo:efficientbasicprocedure} $n$ times
with the same unitary $Y$ will give an estimate $b$ of 
$\cF_{\Lambda}(Y)$ such that
\[
\Pr[|b - \cF_{\Lambda}(Y)| > \frac{\epsilon}{2}] \leq
2 \exp(-\epsilon^2 n / 12) \leq
\frac{\delta}{2},
\]
where the probability arises because of inherent quantum uncertainty of 
measurement outcomes. Overall we get,
$
\Pr[|b - \bcF_{\Lambda}| > \epsilon] \leq
\delta.
$
Algorithm~\ref{algo:4design} takes quantum running time
$O(\epsilon^{-2} \log \delta^{-1} (\log d)^3)$ which
is the same as the zigzag product based approximate $2$-design 
Algorithm~\ref{algo:efficient} of Section~\ref{sec:efficient}. 
However it consumes only
$O(\log d)$ random bits, where the constant hiding in the $O(\cdot)$
is independent of $\epsilon$ and $\delta$, which is less than what 
Algorithm~\ref{algo:efficient} consumes.
The drawback of Algorithm~\ref{algo:4design} is that it requires that
$\frac{108}{\epsilon^2 d} < \frac{\delta}{2}$, whereas
Algorithm~\ref{algo:efficient} has no such assumption. In practice, 
this means that Algorithm~\ref{algo:4design} can only be used if
the dimension $d$ is large and the estimation error $\epsilon$ and 
confidence error $\delta$ are not too small.

\section{Randomness efficient algorithm using approximate $4l$-qTPE}
\label{sec:qTPE}
We now address the main drawback of Algorithm~\ref{algo:4design} i.e.
what if the dimension $d$ is not large enough? This is a pertinent 
question because in the near future we hope to experimentally implement
and benchmark quantum circuits acting on ten to fifty qubits. For 
ten to twenty qubits, the dimension $d$ may not be large enough to 
satisfy the assumption 
$\frac{108}{\epsilon^2 d} < \frac{\delta}{2}$ for reasonably small 
values of $\epsilon$ and $\delta$. It would be nice to have an efficient
algorithm using a few random bits for arbitrary values of $d$, $\epsilon$
and $\delta$. Algorithm~\ref{algo:efficient} is one such algorithm, 
but it requires us to be able to implement the entire plethora of 
$2^{\Omega(\log (d/\epsilon))}$ unitaries in an approximate $2$-design.
In some technological scenarios, it may be better to spend a few more
random bits but reduce the number of unitaries that may need to be 
implemented by the algorithm. Algorithm~\ref{algo:4design} is indeed
a step in this direction but it suffers from the drawback mentioned
above. 

To achieve this goal, we develop Algorithm~\ref{algo:qTPE} by combining 
ideas from 
Algorithms~\ref{algo:efficient} and \ref{algo:4design}. 
Algorithm~\ref{algo:qTPE} is a two-phase algorithm using an
$l$-qTPE for $l \geq 4$, but still small enough to be efficiently
implementable.
\begin{algorithm}[!!!hhh]
\label{algo:qTPE}
\SetAlgoLined
\caption{Randomness efficient algorithm using approximate $4l$-qTPE}
{\bf Input:} $\epsilon, \delta > 0$.

\ 

{\bf Assumption:} 
$\epsilon < \frac{\bcF_{\Lambda}}{2}$,
$4 \log (16/\delta) < \frac{d^{1/6}}{10 \log d}$.

\ 

{\bf Define:} 
$l \triangleq \log (16/\delta)$,
$t \triangleq \lceil \frac{2^{11} l}{\epsilon^2 d} \rceil$,
$\lambda \triangleq (\frac{\epsilon^2}{2^5 d^2})^l$,
$n \triangleq \frac{16 \log (4/\delta^{-1})}{\epsilon^2}$, 
$
\theta \triangleq 
\frac{\delta}{4}(\frac{\epsilon}{n})^{\frac{\epsilon^2 n}{16}}.
$

\ 

{\bf Take:} set $\cV$ to be a $(d, s, \lambda, 4l)$-qTPE
with $\log |\cV| = O(\log \lambda^{-1})$. A unitary of $\cV$ can be
implemented by circuits of size 
$
O((\log d)^2 (\log \delta^{-1})^2 
   (\log \log \delta^{-1}) (\log \lambda^{-1})
).
$
This follows from the zigzag product \cite{Sen}.

\ 

{\bf Choose:} independently and uniformly at random unitaries
$V_1, \ldots, V_t$ from $\cY$. Call this set of unitaries as $\cY$.

\ 

{\bf Construct:} a sequence $Y_1, Y_2, \ldots, Y_n$ of unitaries from
$\cY$ as the output of $f: \{0,1\}^r \rightarrow \cY^n$ guaranteed
by Fact~\ref{fact:sampling}, when a uniformly random input
$z \in \{0,1\}^r$ is fed to $f$, where
$
r \triangleq
\epsilon^2 n \log t  + 2 \log \theta^{-1}.
$

\ 

{\bf Run:} Algorithm~\ref{algo:efficientbasicprocedure} 
with unitaries $Y_1, Y_2, \ldots, Y_n$.
Record the outputs $b_1, b_2, \ldots, b_n$.
Declare $b := n^{-1} \sum_{i=1}^n b_i$ as the estimate for 
$\bcF_{\Lambda}$.
\end{algorithm}

In Phase~1,
we take a $(d, s, \lambda, 4l)$-qTPE $\cV$ and sample from it 
independently and uniformly at random a small 
number of unitaries $V_1, \ldots, V_t$. For
$
l = \log (16 / \delta),
$
$
t = \frac{2^{11} l}{\epsilon^2 d} =
\frac{2^{11} \log (16/\delta)}{\epsilon^2 d}.
$
$
\lambda \triangleq 
(\frac{\epsilon^2}{2^5 d^2})^l =
(\frac{\epsilon^2}{2^5 d^2})^{\log 16 / \delta},
$
Proposition~\ref{prop:conc} will ensure that the
gate fidelity averaged over this small set of unitaries is 
$\epsilon/2$-close to the actual average gate fidelity $\bcF_{\Lambda}$
with confidence at least $1 - \frac{\delta}{2}$. The number of random
bits required to independently sample $t$ times uniformly from $\cV$ is
\begin{equation}
\label{eq:phase1}
O(t \log \lambda^{-1}) = 
O(\epsilon^{-2} d^{-1} (\log \delta^{-1})^2 
  (\log d + \log \epsilon^{-1})
 ).
\end{equation}
Each unitary in $\cV$ can be implemented in quantum time
$
O((\log d)^2 (\log \delta^{-1})^3 
   (\log \log \delta^{-1}) \log (d/\epsilon)
).
$
The constants hiding in both the $O(\cdot)$ notations above are
independent of $d$, $\epsilon$, $\delta$.
These facts follow from the so-called zigzag product \cite{Sen}.
Morally speaking, this preprocessing step
achieves the goal of the first phase of Algorithm~\ref{algo:4design}.

In Phase~2 of our algorithm,
we sample only from the set $\cY \triangleq \{V_1, \ldots, V_t\}$ 
in an approximate limited
independence fashion as in Fact~\ref{fact:sampling}, and
run Algorithm~\ref{algo:efficientbasicprocedure} on those samples in
so as to empiricially estimate the gate fidelity 
averaged over the set $\cY$.
We apply Fact~\ref{fact:sampling} with 
$
n \triangleq
\frac{2^4 \log (4/\delta)}{\epsilon^2},
$
$
\theta \triangleq 
\frac{\delta}{4} (\frac{\epsilon}{n})^{\frac{\epsilon^2 n}{16}} =
\frac{\delta}{4} 
(\frac{\epsilon^3}{2^4 \log (4/\delta)})^{\log (4/\delta)},
$
\begin{equation}
\label{eq:phase2}
\begin{array}{rcl}
r 
& \triangleq &
\epsilon^2 n \log t + 2 \log \theta^{-1} \\
& = &
16 \log \delta^{-1} 
(\log \log \delta^{-1} + 2 \log \epsilon^{-1} - \log d + O(1)) +
2 \log (4/\delta) \\
&   &
{} +
2 \log \delta^{-1} (\log \log \delta^{-1} + 3 \log \epsilon^{-1} + O(1)) \\
& = &
18 \log \delta^{-1} \log \log \delta^{-1} + 
38 \log \delta^{-1} \log \epsilon^{-1} -
16 \log \delta^{-1} \log d + O(\log \delta^{-1}).
\end{array}
\end{equation}
Morally speaking, this achieves the goal of the second
phase of Algorithm~\ref{algo:efficient}. 

We now analyse the running time and number of random bits used
by Algorithm~\ref{algo:qTPE}. The total confidence of Phases~1 and 2
is at least $1 - \delta$.
Suppose
$
\frac{2^{11} \log (16/\delta)}{\epsilon^2 d} \leq 1.
$
Then, we set $t = 1$ and so Phase~2 does not use any more random bits.
The total number of random bits used in Phases~1 and 2 now becomes
$O(\log \delta^{-1} \log (d/\epsilon))$,
where the constant hiding in the $O(\cdot)$ notation does not depend
on $\epsilon$, $\delta$, $d$. This is comparable to the total number of
random bits used by Algorithm~\ref{algo:efficient}. 
However, Algorithm~\ref{algo:qTPE} requires
us to implement only one unitary as opposed to implementing  potentially 
all the unitaries in an approximate $2$-design in  
Algorithm~\ref{algo:efficient}. Also the constraint on the dimension
$d$ is less stringent than the constraint required by 
Algorithm~\ref{algo:4design}.

Suppose
$
\frac{2^{11} \log (16/\delta)}{\epsilon^2 d} > 1.
$
The total number of random bits used in Phases~1 and 2 becomes
\[
O(\epsilon^{-2} d^{-1} (\log \delta^{-1})^2 \log (d/\epsilon) +
  \log \delta^{-1} \log \log \delta^{-1}
 )
\]
where the constant hiding in the $O(\cdot)$ notation does not depend
on $\epsilon$, $\delta$, $d$. Now the number of random bits used
is greater than that of Algorithm~\ref{algo:efficient}. 
However, Algorithm~\ref{algo:qTPE} requires
us to implement only 
$
\frac{\log (16/\delta)}{\epsilon^2 d}
$
unitaries as opposed to implementing  potentially 
$2^{\Omega(\log (d/\epsilon))}$ unitaries in
Algorithm~\ref{algo:efficient}. This is a big saving on the number
of unitaries an algorithm has to potentially implement.

In both cases, the total running time of Algorithm~\ref{algo:qTPE} is
$
O(\epsilon^{-2} (\log \delta^{-1})^4
  (\log d)^2 (\log \log \delta^{-1}) \log (d/\epsilon)
),
$
which is slightly worse than that of Algorithm~\ref{algo:efficient}.
The worse running time is due to the fact that Algorithm~\ref{algo:qTPE}
needs a $O(\log \delta^{-1})$-qTPE whereas Algorithm~\ref{algo:efficient}
only needs a $2$-qTPE. Any future improvement in construction of
$t$-qTPEs for large $t$ will improve the running time of
Algorithm~\ref{algo:qTPE}.

\section{Conclusion}
\label{sec:conclude}
We have described three new algorithms for
efficient in-place estimation, without using ancilla qubits, of 
average fidelity of a quantum logic gate 
using much fewer  random bits than
what was known so far. We considered only in-place algorithms in this
work because qubits are likely to remain an expensive resource 
in experimental implementations in the near future, and so we would like
to avoid ancilla qubits as far as possible. We achieve the 
reductions in the number of
random bits by appealing to two powerful tools. The first tool,
a limited independence pseudorandom generator,
comes from classical derandomisation theory in computer science. 
It is used in the first and the third estimation algorithms. 
The second tool,
an approximate $l$-quantum tensor product expander ($l$-qTPE) for
 moderate values of $l$, is a recent
quantum computational object defined as an analogue of an
approximate $l$-wise independent pseudorandom generator well known
from classical derandomisation theory. In fact, to obtain our parameters 
we actually have to appeal to the state of the art in quantum tensor
product expanders, which in turn were obtained by `quantising' 
another famous result from classical derandomisation theory viz. 
the zigzag product of graphs. The second tool is used in our second,
and more strongly, in our third estimation algorithm.

Each of our algorithms have unique features that, depending upon the
experimental limitations and desired parameters, sometimes make one of
them the most suitable, sometimes another. If one wants an efficient
 algorithm
that works for all values of the gate dimension $d$, estimation error
$\epsilon$ and confidence error $\delta$, and uses the least number
of random bits, then the first algorithm is the way to go. Reducing
the number of random bits will increase the reliability of the
estimation in practice as explained earlier. Moreover, this algorithm
can be chosen to be implemented using only Clifford gates which may be an 
advantage for some technologies. 

If the gate dimension is
large, then the second algorithm is the best one. It uses the least
number of random bits and needs to implement only one unitary from
an approximate $4$-qTPE. If the aim is to use an algorithm that works
for all parameter values but needs to implement as few unitaries as
possible, then the third algorithm is the right one. Though it uses
more random bits than the first algorithm, it needs to implement much
fewer unitaries than the first. This feature can be the most crucial
for certain technologies.

Moving on from estimating average gate fidelity, 
it will be 
interesting to find other applications of both classical derandomisation
as well as quantum derandomisation tools in quantum computation, both
theoretical and experimental.
The experimental implementations of quantum computers of the near future,
of the order of tens of qubits, will be very noisy. 
Reliable and efficient test suites optimised for every bit of
precision and performance  are the need of the hour to measure
progress in the exciting times to come.


\begin{thebibliography}{99}

\bibitem{EmersonAlickiZyczkowski}
{Emerson, C.}, {Alicki, R.} and {\.{Z}yczkowski, K.}
\newblock Scalable noise estimation with random unitary operators.
\newblock {\em Journal of Optics B: Quantum and Semiclassical Optics}, 
{\bf 7}, S347--S352, (2005).

\bibitem{DankertCleveEmersonLivine} 
{Dankert, C.}, {Cleve, R.}, {Emerson, J.} and {Livine, E.}
\newblock Exact and approximate unitary 2-designs and their
application to fidelity estimation. 
\newblock {\em Physical Review A}, {\bf 80}, 012304, (2009).

\bibitem{EmersonEtAl} 
{Emerson, J.}, {Silva, M.}, {Moussa, O.}, {Ryan, C.},
{Laforest, M.}, {Baugh, J.}, {Cory, D.} and {Laflamme, R.}
\newblock Symmetrized Characterization of Noisy Quantum Processes.
\newblock {\em Science}, {\bf 317}, 1893--1896, (2007).

\bibitem{BellareRompel}
{Bellare, M.} and {Rompel, J.}
\newblock Randomness-efficient oblivious sampling.
\newblock {\em 35th Annual IEEE Symp. on Found. of Comp. Sc.},
               276--287, (1994).

\bibitem{Low}
{Low, R.}
\newblock Large Deviation Bounds for $k$-designs.
\newblock {\em Proc. of the Royal Society A},
               {\bf 465(2111)}, 3289--3308, (2009).

\bibitem{Sen}
{Sen, P.}
\newblock Efficient quantum tensor product expanders and unitary 
          $t$-designs via the zigzag product.
\newblock Arxiv:1808.10521, (2018).

\bibitem{AlonSpencer}
{Alon, N.} and {Spencer, J.}
\newblock {\em The probabilistic method}.
\newblock 4th Edition, Wiley, (2016).

\bibitem{MagesanBlumekohoutEmerson}
{Magesan, E.}, {Blume-Kohout, R.} and {Emerson, J.}
\newblock Gate fidelity fluctuations and quantum process invariants.
\newblock {\em Physical Review A}, {\bf 84}, 012309, (2011).

\bibitem{DiVincenzoLeungTerhal}
{DiVincenzo, D.}, {Leung, D.} and {Terhal, B.}
\newblock Quantum data hiding.
\newblock {\em IEEE Trans. Inf. Theory}, {\bf 48(3)}, 580--599, (2002).

\bibitem{CleveLeungLiuWang}
{Cleve, R.}, {Leung, D.}, {Liu, L.} and {Wang, C.}
\newblock Near-linear constructions of exact unitary $2$-designs,
\newblock {\em Quant. Inf. Comp.}, {\bf 16(9\&10)}, 0721--0756, (2016).

\bibitem{AlonGoldreichHastadPeralta}
{Alon, N.}, {Goldreich, O.}, {H{\aa}stad, J.} and {Peralta, R.}
\newblock Simple Construction of Almost k-wise Independent 
          Random Variables.
\newblock {\em Rand. Struct. Alg.}, {\bf 3(3)}, 289--304, (1992).

\bibitem{SchmidtSiegelSrinivasan}
{Schmidt, J.}, {Siegel, A.} and {Srinivasan, A.}
\newblock {Chernoff-–Hoeffding} bounds for applications with 
          limited independence.
\newblock {\em SIAM J. Discrete Math.}, {\bf 8(2)}, 223--250, (1995).

\bibitem{FawziHaydenSen}
{Fawzi, O.}, {Hayden, P.} and {Sen, P.}
\newblock From low-distortion norm embeddings to explicit uncertainty 
          relations and efficient information locking.
\newblock {\em J. ACM.}, {\bf 60(6)}, 44:1--44:60, (2013).

\bibitem{KitaevWatrous}
{Kitaev, A.} and {Watrous, J.}
\newblock Parallelization, amplification, and exponential time 
          simulation of quantum interactive proof systems.
\newblock {\em Proc. 32nd Annual ACM Symp. on Theory of Computing}, 
          608--617, (2000).

\bibitem{HarrowHastings}
{Harrow, A.} and {Hastings, J.}
\newblock Classical and Quantum Tensor Product Expanders.
\newblock {\em Quant. Inf. and Comp.}, {\bf 9}, 336:1--336:18, (2009).

\bibitem{Goldreich}
{Goldreich, O.}
\newblock {\em Pseudorandom Generators: A Primer.}
\newblock Univ. Lect. Ser., {\bf 55}, Am. Math. Soc., (2010).

\bibitem{Vershynin}
{Vershynin, R.}
\newblock {\em High dimensional probability.}
\newblock Cambridge Univ. Press, (2018).

\end{thebibliography}
\end{document}